\documentclass[final,12pt]{colt2021} 

\title[Approximation Algorithms for Socially Fair Clustering]{Approximation Algorithms for Socially Fair Clustering}
\usepackage{times}

\usepackage{algorithm}
\usepackage{algorithmic}

\newtheorem{claim}[theorem]{Claim}
\newtheorem{observation}[theorem]{Observation}

\def\floor#1{\lfloor {#1} \rfloor}

\def\script#1{\mathcal{#1}}

\def\sep{\;|\;}

\def\E{\mathbf{E}}
\def\Var{\mathbf{Var}}

\def\sB{\script{B}}

\def\nn{\mathrm{NN}}
\def\cllp{ClusterLP}
\def\optcllp{BasicClusterLP}

\newcommand\faircost{{\operatorname{fair-cost}}}
\newcommand\cost{{\operatorname{cost}}}
\newcommand\vol{{\operatorname{vol}}}
\def\R{\script{R}}

\def\g{{\boldsymbol{g}}}

 \coltauthor{\Name{Yury Makarychev} \Email{yury@ttic.edu}\and
  \Name{Ali Vakilian} \Email{vakilian@ttic.edu}\\
  \addr Toyota Technological Institute at Chicago, 6045 S Kenwood Ave, Chicago, IL 60637, USA}

\begin{document}

\date{}

\maketitle

\begin{abstract}
We present an $(e^{O(p)} \frac{\log \ell}{\log\log\ell})$-approximation algorithm for socially fair clustering with the $\ell_p$-objective.
In this problem, we are given a set of points in a metric space.
Each point belongs to one (or several) of $\ell$ groups.
The goal is to find a $k$-medians, $k$-means, or, more generally, $\ell_p$-clustering that is simultaneously good for all of the groups. More precisely, we need to find a set of $k$ centers $C$ so as to minimize the maximum over all groups $j$ of $\sum_{u \text{ in group }j} d(u,C)^p$.

The socially fair clustering problem was independently proposed by~\cite*{ghadiri2020fair} and \cite*{abbasi2020fair}. Our algorithm improves and generalizes  their $O(\ell)$-approximation algorithms for the problem. 

The natural LP relaxation for the problem has an integrality gap of $\Omega(\ell)$. In order to obtain our result, we introduce a strengthened LP relaxation and show that it has an integrality gap of $\Theta(\frac{\log \ell}{\log\log\ell})$ for a fixed~$p$. Additionally, we present a bicriteria approximation algorithm, which generalizes the bicriteria approximation of~\cite{abbasi2020fair}.
\end{abstract}

\section{Introduction}\label{sec:intro}
Due to increasing use of machine learning in decision making, there has been an extensive line of research on the societal aspects of algorithms~\citep{galindo2000credit,kleinberg2017inherent,chouldechova2017fair,dressel2018accuracy}. The goal of research of this area is to understand, on one hand, \emph{what it means for an algorithm to be fair and how to formally define the fairness requirements} and, on the other hand, \textit{how to design efficient algorithms that meet the fairness requirements}. By now, there have been a rich amount of studies on both fronts for different machine learning tasks~\citep{dwork2012fairness,feldman2015certifying,hardt2016equality,chierichetti2017fair,kleindessner2019fair,jung2019center,har2019near,aumuller2020fair}. 
We refer the reader to~\citep{chouldechova2018frontiers,kearns2019ethical} for an overview of different notions of fairness and their computational aspects.

In this work, we study clustering under the notion of group fairness called {\em social fairness} or {\em equitable group representation}, introduced independently by~\citeauthor{abbasi2020fair} and~\citeauthor{ghadiri2020fair}. In this framework, we need to cluster a
dataset with points coming from a number of groups. 
\citeauthor{ghadiri2020fair} observed that both standard clustering algorithms and existing fair clustering algorithms (e.g., those that ensure that various groups are fairly represented in all clusters; see~\citep{chierichetti2017fair}) incur higher clustering costs for certain protected groups (e.g., groups that are defined by a sensitive attribute such as gender and race). This lack of fairness has motivated the study of clustering that minimizes 
the maximum clustering cost across different demographic groups. 
The objective of socially fair clustering was previously studied by~\cite*{anthony2010plant} in the context of {\em robust clustering}. In this setting, a set of possible scenarios is provided and the goal is to find a solution that is simultaneously good for all scenarios.       
\begin{definition}[Socially Fair $\ell_{p}$-Clustering] \label{def:social-fair-clustering}
We are given a metric space $(P, d)$ on $n$ points, $\ell$ groups of points $P_1, \cdots, P_{\ell} \subseteq P$, and the desired number of centers $k$.
Additionally, we are given non-negative demands or weights $w_j: P_j \rightarrow \mathbb{R}$ for points in each group $P_j$. We define the $\ell_p$-cost of group $P_j$ w.r.t. a set of centers $C \subseteq P$ as
$$\cost(C, w_j) = \sum_{u\in P_j} w_j(u) \cdot d(u, C)^p,$$
where $d(u, C) = \min_{c\in C} d(u, c)$. In socially fair $\ell_{p}$-clustering, the goal is to pick a set of $k$ centers $C \subseteq P$  so as to minimize $\max_{j\in[\ell]} \cost(C, w_j)$, which we call the fair cost:
\begin{align}\label{eq:clustering-cost}
\faircost(C) = \faircost(C, \{w_j\}_{j\in [\ell]}):= \max_{j\in [\ell]} \sum_{u\in P_j} w_j(u) \cdot d(u, C)^p.
\end{align}
\end{definition}

The problem was first studied in the context of robust clustering.  \citeauthor{anthony2010plant} introduced it under the name ``robust $k$-medians'' (for $p=1$) and designed an $O(\log |P| + \log \ell)$-approximation algorithm for it (for $p=1$).
Recently, \citeauthor{abbasi2020fair}\footnote{\citeauthor{abbasi2020fair} call this problem \textit{fairness under equitable group representation}.} and \citeauthor{ghadiri2020fair} 
considered this problem in the context of social fairness.
They studied this problem for the most interesting cases, when $p\in\{1,2\}$ and $p =2$, respectively; the papers assumed that $w_j = 1/|P_j|$. In these special cases, the problem is called socially fair $k$-medians ($p=1$) and socially fair $k$-means ($p=2$). Both papers presented $O(\ell)$-approximation algorithms for the variants of fair clustering they study. 
Further, \citeauthor{abbasi2020fair} gave a bicriteria $(2/\gamma, 1/(1-\gamma))$-approximation algorithm for fair $k$-medians and $k$-means. Their algorithm finds a solution with at most $k/(1-\gamma)$ centers, whose cost is at most $2/\gamma$ times the optimal cost for $k$ centers. Also, \citeauthor{ghadiri2020fair} designed a fair variant of Lloyd's heuristic for fair $k$-means clustering. 
Further, \cite*{bhattacharya2014new} showed that it is hard to approximate socially fair $k$-medians by a factor better than $\Omega(\frac{\log \ell}{\log \log \ell})$ unless $\mathrm{NP} \subseteq \cap_{\delta>0} \mathrm{DTIME}(2^{n^\delta})$. This hardness result holds even for uniform and line metrics.
We remark that since the result by~\citeauthor{bhattacharya2014new} holds for uniform metrics -- metrics where all distances are either $0$ or $1$ -- the same $\Omega(\log \ell / \log\log \ell)$-hardness result also applies to socially fair clustering with other values of $p>1$.

\paragraph{Our Results.}
In this paper, we consider socially fair clustering for arbitrary $p\in[1,\infty)$ and arbitrary demands/weights. We do not require groups  $P_j$ to be disjoint. 

Our main contribution is an $(e^{O(p)} \frac{\log \ell}{\log\log \ell})$-approximation algorithm for socially fair $\ell_p$-clustering problem. For socially fair $k$-means, our algorithm improves the $O(\ell)$-approximation algorithms by~\citeauthor{abbasi2020fair} and~\citeauthor{ghadiri2020fair} to $O(\frac{\log\ell}{\log\log\ell})$. For socially fair $k$-medians, our algorithm improves the $O(\log n + \log \ell)$-approximation algorithm by~\citeauthor{anthony2010plant} to $O(\frac{\log\ell}{\log\log\ell})$.
The hardness result by~\citeauthor{bhattacharya2014new} shows that our approximation guarantee for socially fair $k$-clustering with $p=O(1)$, which includes $k$-medians and $k$-means, is optimal up to a constant factor unless $\mathrm{NP} \subseteq \cap_{\delta>0} \mathrm{DTIME}(2^{n^\delta})$.
We also give an $(e^{O(p)}/\gamma, 1/(1-\gamma))$-bicriteria approximation for the socially fair $\ell_p$-clustering problem (where $\gamma \in (0, 1)$). This result generalizes the bicriteria guarantee of~\citeauthor{abbasi2020fair} for the $k$-means and $k$-medians objectives to the case of arbitrary $p \geq 1$ and arbitrary demands $w_j$. 

\begin{theorem}\label{thm:main-approx}
There exists a polynomial-time algorithm that computes an $(e^{O(p)}\frac{\log \ell}{\log\log \ell})$-approximation for the socially fair $\ell_p$-clustering problem (where $\ell$ is the number of groups).
\end{theorem}

\begin{theorem}\label{thm:bicriteria-approx}
There exists an algorithm that computes a bicriteria $(e^{O(p)}/\gamma, 1/(1-\gamma))$-approximation for the socially fair $\ell_p$-clustering problem (where $\gamma\in (0,1)$). The algorithm finds a solution with at most $k/(1-\gamma)$ centers, whose cost is at most $e^{O(p)}/\gamma$ times the optimal cost for $k$ centers.
\end{theorem}

Below, we will mostly focus on proving Theorem~\ref{thm:main-approx}. We prove Theorem~\ref{thm:bicriteria-approx} in Section~\ref{sec:bicriteria}.
Our algorithms are based on linear programming.
However, as shown by~\citeauthor{abbasi2020fair}, the integrality gap of the natural LP relaxation for the socially fair $\ell_p$-clustering is $\Omega(\ell)$.
In order to get an approximation factor better than $\Theta(\ell)$, we strengthen the LP by introducing an extra set of constraints. Loosely speaking, new constraints require that each point $u$ be connected only to centers at distance at most $\Delta$ from $u$, where the value of $\Delta$ depends on the specific point $u$ (we discuss these constraints in detail below). Once we solve the LP relaxation with additional constraints, we apply the framework developed by
\cite*{charikar2002constant} for the $k$-medians problem. Using their framework, we transform the problem instance and the LP solution. We get an instance with a set of points $P'$ and an LP solution that ``fractionally'' opens a center at every point $u\in P'$; specifically, each point $u\in P'$ will be at least $(1-\gamma)$-fractionally open where $\gamma=1/10$.
Now to obtain our approximation results, we use independent sampling in combination with some techniques from \cite{charikar2002constant}. The analysis crucially uses the LP constraints we introduced. Our bicriteria approximation algorithm simply outputs set $C=P'$.

\paragraph{Related Work.}
Clustering has been an active area of research in the domain of fairness for unsupervised learning. One notion of group fairness for clustering, introduced by \cite{chierichetti2017fair}, requires that output clusters are balanced. This notion of fairness has been extended in a series of papers~\citep{abraham2019fairness,bercea2019cost,bera2019fair,schmidt2019fair,backurs2019scalable,ahmadian2019clustering,huang2019coresets}. Another well-studied notion of group fairness for clustering requires the chosen centers fairly represent the underlying population. Various aspects of clustering under this notion of fairness have been studied in the literature~\citep{hajiaghayi2010budgeted,krishnaswamy2011matroid,chen2016matroid, krishnaswamy2018constant,kleindessner2019fair,chiplunkar2020solve,jones2020fair}. 
There has been also extensive research on other notions of fairness for clustering~\citep{chen2019proportionally,jung2019center,mahabadi2020individual,micha2020proportionally,kleindessner2020notion,brubach2020pairwise,anderson2020distributional}. 

The $k$-clustering problem with the $\ell_p$-objective $\sum_{u\in P} d(u, C)^p$, which is a natural generalization of $k$-medians, $k$-means, and $k$-center, is a special case of socially fair $\ell_p$-clustering when the number of groups $\ell$ is equal to one. As observed by~\cite{chakrabarty2019approximation}, a slightly modified variants of the classic algorithms by~\cite{charikar2002constant} and~\cite{jain2001approximation} for $k$-medians gives an $e^{O(p)}$-approximation for $k$-clustering with the $\ell_p$-objective. We remark that our algorithm for socially fair $\ell_p$-clustering also gives $e^{O(p)}$ approximation when $\ell=1$.   

\section{Preliminaries}
We denote the distance from point $u$ to set $C$ by 
$$d(u, C) = \min_{v\in C} d(u,v).$$
To simplify notation, we will assume below that weights $w_j$ are defined on the entire set $P$, but $w_j(u) = 0$ for $u\notin P_j$. We denote $w(v) = \sum_{j\in[\ell]} w_j(v)$. We will use the following definitions and results in the paper.

\begin{definition}[approximate triangle inequality]\label{def:approx-tri-ineq}
A distance function satisfies the $\alpha$-approximate triangle inequality over a set of points $P$ if, $\forall u,v,w\in P, d(u,w) \leq \alpha \cdot (d(u,v) + d(v,w))$
\end{definition}

\begin{claim}[Corollary A.1 in~\citep{makarychev2019performance}]\label{clm:tri-ineq}
Let $(P, d)$ be a metric space. Consider distance function $d(u,v)^p$. It satisfies the $\alpha_p$-approximate triangle inequality for $\alpha_p = 2^{p-1}$.
\end{claim}

\begin{theorem}[Bennett's Inequality~\citep{bennett1962}; also see Theorem 2.9.2 in~\citep{vershynin2018high}]\label{thm:Bennett}
Let $X_1,\dots,X_n$ be independent mean-zero random variables and $S=\sum_{i=1}^n X_i$. Assume that (a) $\Var[S]\leq \sigma^2$ and (b) for all $i\in [n]$, $|X_i|\leq M$ always. Then, for every $t\geq 0$, we have
$$
\Pr[S\geq t]\leq \exp \left(-\frac{\sigma^2}{M^2} h \left( \frac{tM}{\sigma^2} \right)\right)
\leq \exp\left(-\frac{t}{2M}\log\left(\frac{tM}{\sigma^2} + 1\right)\right),
$$
for $h(x) = (1 + x) \ln(1 + x) - x^2 \geq \frac{1}{2} \cdot x \log (x+1)$.
\end{theorem}

\section{LP Relaxations for Socially Fair \texorpdfstring{$\ell_{p}$-}{}Clustering}\label{sec:lp-gen}
In this section, we describe an LP-relaxation for the socially fair $\ell_p$-clustering.
We start with a natural LP relaxation for the problem studied by~\cite{abbasi2020fair}. The relaxation is a generalization of the standard LP for $k$-means and $k$-medians clustering \citep{charikar2002constant}: 
For every $v\in P$, we have an LP variable $y_v$ that denotes whether $v$ belongs to the set of selected centers (in an integral solution, $y_v =1$ if $v\in C$ and $0$ otherwise);
for every $u, v\in P$, we have an LP variable $x_{uv}$ that denotes whether $v$ is the closest center to $u$ in the selected set of centers $C$.
\begin{align*}
\textbf{LP Relaxation: }&\optcllp(\{w_j\}) \\[1mm]
\text{minimize }&\ \max_{j\in[\ell]} \sum_{u\in P_j, v\in P} w_j(u) \cdot d(u, v)^p \cdot x_{uv} \\[1mm]
\text{s.t.}\qquad&\sum_{v\in P} x_{uv} = 1 && \forall u\in P\\
&\sum_{v\in P} y_v \leq k \\
&x_{uv} \leq y_{v} && \forall u,v\in P\\
&x_{uv}, y_{u} \geq 0 && \forall u,v\in P
\end{align*}
Note that the objective is not linear as written. However, we can rewrite this relaxation as a true LP by introducing a new variable $A$, adding LP constrains 
$A \geq \sum_{u\in P_j, v\in P} w_j(u) \cdot d(u, v)^p \cdot x_{uv}$ for each $j$, and then minimizing $A$ in the objective.

As was shown by~\citeauthor{abbasi2020fair}, this LP relaxation has an integrality gap of $\Omega(\ell)$ (for all $p\in[1,\infty)$). As discussed in the Introduction, we strengthen this LP by introducing an extra set of constraints. To describe these constraints, we need some notation. For each point $v\in P$, we denote the ball of radius $r\geq 0$ around $v$ by $\sB(v, r) := \{u\in P \sep d(v,u) \leq r\}$. 
We define the volume of a ball $\sB(v, r)$ as $\vol_v(r) = \max_{j\in [\ell]} \sum_{u \in \sB(v, r)} w_j(u) \cdot r^{p}$.
\begin{remark} We use this definition of $\vol_v(r)$ so as to ensure that the following property holds. Consider a set of centers $C$ and assume that the distance from all points in $\sB(v,r)$ to $C$ is at least $r$. Then the cost of solution $C$ is at least $\vol_v(r)$.
\end{remark}
Clearly, function $\vol_v(r)$ is a non-decreasing function of $r$ that goes to infinity as $r\to \infty$ (unless all demands $w_j$ are identically 0). 
Further, if there are no points at distance exactly $r$ from $v$, then $\vol_v(r)$ is continuous at $r$; otherwise, it is right-continuous at $r$
and may or may not be left-continuous. Denote $\vol_v(r-0) = \lim_{s\to r - 0} \vol_v(s)$ (the limit of $\vol_v$ at $r$ from the left); we let $\vol_v(0-0) = 0$.
Given a {\em budget} parameter $z$, define $\Delta_z(v) = \min\{r: \vol_v(r) \geq z\}$.
Note that $\vol_v(\Delta_z(v) - 0) \leq z \leq \vol_v(\Delta_z(v))$.
In this paper, we will always compute $\vol_v(r)$ and $\Delta_z(u)$ with respect to the original demands $w_j$ (not demands $w'_j$ we define later).
Claim~\ref{claim:up-bd-vaild} explains why we consider $\Delta_z(v)$.
\begin{claim}\label{claim:up-bd-vaild}
Consider an instance of the socially fair $\ell_p$-clustering problem. Denote its optimal cost by $z^*$.  Let $C^*$ be an optimal set of centers and $z \geq z^*$.
Then for all $v\in P$ with $w(v) = \sum_{j\in[\ell]} w_j(v) > 0$, we have $d(v, C^*) \leq 2\Delta_{z}(v)$.
\end{claim}
\begin{proof}
Assume to the contrary that there exists a point $v$ such that $d(v, C^*) > 2\Delta_{z}(v)$. Denote $R = \Delta_{z}(v)$.
Then, for every $u\in \sB(v, R)$, we have $d(u, C^*) \geq d(v, C^*) - d(u,v) >R$.

If $z > 0$, choose $j$ so that $\sum_{u \in \sB(v, R)} w_j(u) \cdot R^{p} = \vol_v(R) \geq z$. 
If $z = 0$, choose $j$ so that $w_j(v) > 0$. Note that in either case, $S := \sum_{u \in \sB(v, R)} w_j(u) > 0$.
We have,
\begin{align*}
\faircost(C^*) &\geq \cost(C^*, w_j) \geq \sum_{u \in \sB(v, R)} w_j(u) d(u, C^*)^p \\
&\substack{\tiny\text{since } S > 0 \\ >} \sum_{u \in \sB(v, R)} w_j(u) R^p \geq z \geq z^* = \faircost(C^*).    
\end{align*}
We get a contradiction.
\end{proof}
We now state our strengthened LP relaxation, which has a new family of constraints (\ref{const:up-bd}). 

\begin{align}
\textbf{LP Relaxation: }&\cllp(\{w_j\}, \{\Delta_z(u)\}_{u\in P}, \lambda)\nonumber\\[1mm]
\text{minimize }& \ \rlap{$\max_{j\in[\ell]} \sum_{u\in P_j, v\in P} w_j(u) \cdot d(u, v)^p \cdot x_{uv}$} \nonumber\\[1mm]
\text{s.t.}\qquad &\sum_{v\in P} x_{uv} = 1 && \forall u\in P \label{cst:point-assignment}\\
&\sum_{v\in P} y_v \leq k \label{cst:center-count}\\
&x_{uv} \leq y_{v} && \forall u,v\in P \label{cst:center-cap}\\
&x_{vu} = 0 &&\forall u,v \text{ s.t. } w(v) > 0 \text{ and } d(u,v) > \lambda \Delta_z(v) \label{const:up-bd}\\
&x_{uv}, y_{u} \geq 0 && \forall u,v\in P \label{cst:positive}
\end{align}

It follows from Claim~\ref{claim:up-bd-vaild} that this is a valid relaxation if $\lambda\geq 2$ and $z$ is at least $z^*$ (the cost of the optimal fair clustering). We remark that the bicriteria approximation can be achieved from rounding $\optcllp$; however, for the simplicity of exposition, we also use $\cllp$ in our
bicriteria approximation algorithm.

We show that the integrality gap of $\cllp$ is $\Omega(\log \ell/\log\log \ell)$ and thus our rounding procedure is essentially optimal.
\begin{theorem}\label{thm:gap}
The integrality gap of $\cllp$ is 
$\Omega(\log \ell/\log\log \ell)$.
\end{theorem}
We make the statement of Theorem~\ref{thm:gap} precise and prove it in Appendix~\ref{sec:int_gap}.
\paragraph{Outline of Our Algorithms.} We provide two approximation guarantees for fair $\ell_p$-clustering.
The first algorithm gives an $(e^{O(p)}\frac{\log \ell }{\log\log \ell})$-approximation and the second one gives a bicriteria $(e^{O(p)}/\gamma, 1/(1-\gamma))$-approximation.

We will assume below that we have an approximation $z_\g$ for the cost $z^*$ of the optimal solution such that $z^* \leq z_\g \leq 2z^*$. A standard argument shows that we can do that: we can simply run our algorithm with various values of $z_\g$ and, loosely speaking, output the best clustering the algorithm finds. We formally explain why we can assume that we have such a value of $z_\g$ in Appendix~\ref{sec:appendix-guess-z}.

Our algorithms consist of two steps. In the first step, which is common for both our algorithms, we construct a $(1-\gamma)$-restricted solution of $\cllp$, defined as follows.

\begin{definition}\label{def:restricted-sol}
A solution $(x,y)$ of $\cllp(\{w_j\}, \{\Delta_z(u)\}_{u\in P}, \lambda)$ is called {\em $(1-\gamma)$-restricted w.r.t. $P'\subseteq P$} if $y_u \geq 1 -\gamma$ for all $u\in P'$ and $y'_u = 0$ for all 
$u\notin P'$. 
\end{definition}
We find a $(1-\gamma)$-restricted solution by closely following the approach by~\cite{charikar2002constant}. 
However, as we work with a different objective function and LP,
some careful modifications to the approach by~\citeauthor{charikar2002constant} are required.
The second steps for our approximation and bicriteria approximation algorithms are different. The former uses an independent randomized rounding in combination with some techniques from \citeauthor{charikar2002constant}; the latter algorithm does not actually need any rounding -- it  simply outputs $C=P'$.

\section{Step 1: Constructing \texorpdfstring{$(1-\gamma)$}{(1-gamma)}-Restricted Solutions}\label{sec:step-1}
\subsection{Overview}
In this section, we present a polynomial-time algorithm that given an instance ${\cal I}$ with demand functions $\{w_j\}$ constructs an instance ${\cal I}'$ with new demands $\{w'_j\}$ on the same set of points $P$. Let $P'$ be the support of $\{w'_j\}$: $P' = \{u: w'_j(u)\neq 0 \text{ for some } j\in [\ell]\} $. The new instance will satisfy the following properties (which we now state informally).
\begin{enumerate} 
\item All points in $P'$ are well separated (the distance between every two points in $P'$ is ``large''). 
\item Let $(x,y)$ be an optimal LP solution for $\cllp(\{w_j\}, \{\Delta_{z_{\g}}(u)\}_{u\in P}, 2)$. The LP cost of $(x,y)$ w.r.t. new demands $\{w_j'\}$ is at most that w.r.t. the original demands
$\{w_j\}$.
\item  The cost of every solution $C$ w.r.t. original demands $\{w_j\}$ is at most a constant factor greater than the cost of $C$ w.r.t.
demands $\{w'_j\}$ (for every fixed $p$).
\end{enumerate}
After we show how to transform $\cal I$ to ${\cal I}'$ and prove that ${\cal I}'$ satisfies properties (1)-(3), we describe how to convert the optimal LP solution $(x,y)$ for $\cal I$ to a solution $(x',y')$ for ${\cal I}'$, which is $(1-\gamma)$-restricted w.r.t.~$P'$ (see Definition~\ref{def:restricted-sol}).

Now we observe that it is sufficient to design a ``good'' rounding scheme only for $(1-\gamma)$-restricted LP solutions: we use the transformations discussed above, then apply the rounding scheme to the LP solution $(x', y')$ for ${\cal I}'$, obtain a ``good'' solution $C$ for ${\cal I}'$, and then output $C$ as a solution for ${\cal I}$.


\subsection{Consolidating Locations}\label{sec:consolidating-locations}
In this subsection, we describe how we transform instance $\cal I$ to instance ${\cal I}'$.
We call this step location consolidation.
Let $(x,y)$ be an optimal LP solution for $\cllp(\{w_j\}, \{\Delta_{z_{\g}}(u)\}_{u\in P}, 2)$ where $z_{\g} \in [z^*, 2z^*]$ and $z^*$ is the cost for the optimal integral solution of the fair clustering instance.
To describe the location consolidation algorithm, we need to define a notion of ``fractional distance of point $u$ to the center'' according to $(x, y)$. 
We define the fractional distance $\R(u)$ for $u\in P$ as follows:
\begin{align}\label{eq:R-def}
\R(u) := \left(\sum_{v\in P} d(u, v)^p \cdot x_{u v}\right)^{1/p}
\end{align}
Note that if $(x,y)$ is an integral solution, then $\R(u)$ is simply the distance from $u$ to the center of the cluster $u$ is assigned to.

\begin{claim}\label{clm:total-R}
For each group $j\in [\ell]$, we have
$\sum_{u\in P} w_j(u) \cdot \R(u)^p \leq z^*$.
\end{claim}
\begin{proof}
Consider an arbitrary group $j\in [\ell]$.
$$
\sum_{u\in P} w_j(u) \cdot \R(u)^p = \sum_{u,v\in P} w_j(u) \cdot  d(u, v)^p \cdot x_{u v} \leq z^*
$$
where the last inequality holds, since $(x,y)$ is an optimal solution for $\cllp(\{w_j\}, \{\Delta_{z_\g}(u)\}_{u\in P}, 2)$ and 
$\cllp(\{w_j\}, \{\Delta_{z_\g}(u)\}_{u\in P}, 2)$ is a relaxation for fair $\ell_p$ clustering.
\end{proof}



Algorithm~\ref{alg:consolidating_locations} performs location consolidation.   
\begin{algorithm}[h]
	\begin{algorithmic}[1]
		\STATE {\bfseries Input:} $(x, y)$ is an optimal solution of $\cllp(\{w_j\},  \{\Delta_{z_{\g}}(u)\}_{u\in P}, 2)$
		\STATE $\R(v) = \left(\sum_{u\in P}  d(v,u)^p \cdot x_{vu}\right)^{1/p}$ \textbf{for all} $v\in P$		
        \STATE $w'_j(v) = w_j(v)$ \textbf{for all} $v\in P$ and $j\in [\ell]$
		\STATE {\bf sort} the points in $P$ so that $\R(v_1) \leq \R(v_2) \leq \cdots \leq \R(v_n)$
		\FOR{$i = 1$ to $n-1$}
			\FOR{$j = i+1$ to $n$}
				\IF{$d(v_i, v_j) \leq \frac{2}{\gamma^{1/p}} \R(v_j)$  and $\sum_{c\in[\ell]}w'_c(v_i) > 0$}
					\STATE $w'_{t}(v_i) = w'_{t}(v_i) + w'_{t}(v_j)$ \textbf{for all} $t\in [\ell]$
                    \STATE $w'_{t}(v_j) = 0$  \textbf{for all} $t\in [\ell]$
				\ENDIF
			\ENDFOR
		\ENDFOR	
    \end{algorithmic}
	\caption{Consolidating locations.}
	\label{alg:consolidating_locations}
\end{algorithm}
After we initialize variables (lines 2--3), we sort all points in $P$ according to their fractional distance in a non-decreasing order -- $v_1, \cdots, v_n$ -- so that $\R(v_1) \leq \R(v_2) \leq \cdots \leq \R(v_n)$ (line 4).
Then, we consider the points in this order one-by-one. When processing a point $v_i$ with non-zero demand, we check whether there exists another point $v_j$ with non-zero demand such that $j>i$ and $d(v_i, v_j) \leq \frac{2}{\gamma^{1/p}} \R(v_j)$ (line 7). If there is such a point,  we add the demands of $v_j$ to $v_i$ and set the demands of $v_j$ to zero (lines 8-9). 
When the algorithm runs the described procedure,
we say that it \textit{moves} the demand of $v_j$ to $v_i$. Note that after the algorithm processes $v_i$, it never moves the demand of $v_i$ to another point. Thus, the demand of each point can be moved at most once.

We run Algorithm~\ref{alg:consolidating_locations} on optimal LP solution $(x,y)$ and obtain a new set of demands $\{w'_j\}$. Let $w'(v) = \sum_{j\in[\ell]} w'_j(v)$, and $P'$ be the support of $w'$.
Claim~\ref{clm:well-separated} shows that all points in $P'$ are well-separated.
\begin{claim}\label{clm:well-separated}
For every pair of $u,v \in P'$, $d(u,v) > \frac{2}{\gamma^{1/p}} \cdot \max(\R(u), \R(v))$.
\end{claim}
\begin{proof}
The proof follows from the consolidation rule of the algorithm. Suppose that $u$ comes before $v$ in the ordering considered by Algorithm~\ref{alg:consolidating_locations}. Then, since both $u,v$ have non-zero demands and the algorithm has not moved the demand of $v$ to $u$ at the time it processed $u$, $d(u,v) > \frac{2}{\gamma^{1/p}} \R(v) = \frac{2}{\gamma^{1/p}} \max(\R(u),\R(v))$.
\end{proof}
Finally, we show that every solution $(x,y)$ for the original LP relaxation is also a feasible solution with the same or smaller cost for the LP with new demands $\{w_j'\}$.
\begin{lemma}\label{lem:consolidating_locations}
Let $(x, y)$ be a feasible solution for $\cllp(\{w_j\},\{\Delta_z(u)\}_{u\in P}, \lambda)$ with cost $z'$. Then $(x, y)$ is a feasible solution for $\cllp(\{w'_j\}, \{\Delta_z(u)\}_{u\in P},\lambda)$ with cost at most $z'$. 
\end{lemma}
\begin{proof}
In relaxations $\cllp(\{w_j\}, \{\Delta_z(u)\}_{u\in P}, \lambda)$ and $\cllp(\{w'_j\}, \{\Delta_z(u)\}_{u\in P}, \lambda)$, all constraints other than \eqref{const:up-bd} do not depend on demands and thus are the same in both relaxations.
Observe that if $w(v) = 0$ then also $w'(v) = 0$. Therefore, if constraint \eqref{const:up-bd} is present in linear program $\cllp(\{w'_j\}, \{\Delta_z(u)\}_{u\in P}, \lambda)$ for some $u$ and $v$, then it is also 
present in $\cllp(\{w_j\}, \{\Delta_z(u)\}_{u\in P}, \lambda)$ for the same $u$ and $v$ (note that we use the same $\Delta_z(u)$ in both LPs). We conclude that the set of constraints of $\cllp(\{w'_j\}, \{\Delta_z(u)\}_{u\in P}, \lambda)$ is a subset of those of $\cllp(\{w_j\}, \{\Delta_z(u)\}_{u\in P}, \lambda)$.
Thus, since $(x,y)$ is a feasible solution for $\cllp(\{w_j\}, \{\Delta_z(u)\}_{u\in P}, \lambda)$, it is also a feasible solution for $\cllp(\{w'_j\}, \{\Delta_z(u)\}_{u\in P}, \lambda)$.

Now we show that the LP cost of every group $j$ does not increase. The LP cost of group $j$ is $\sum_{u \in P} w'_j(u) \R(u)$. When Algorithm~\ref{alg:consolidating_locations} initializes $w_j'$ in the very beginning (see line 3 of the algorithm),
we have $w'_j(u) = w_j(u)$ for all $u$, and thus at that point $\sum_{u \in P} w'_j(u) \R(u) = \sum_{u \in P} w_j(u) \R(u) \leq z'$.
When the algorithm moves demand from $u$ to $v$ on lines 8-9, we always have $R(v) \leq R(u)$.
Thus, every time lines 8-9 are executed, the value of expression $\sum_{u \in P} w'_j(u) \R(u)$ may only go down.
Therefore, when the cost of LP solution $(x,y)$ with respect to demands $w'_j$ returned by the algorithm is at most $z'$.
\end{proof}

\subsection{Consolidating Centers}\label{sec:consolidating-centers}
In the previous section, we constructed an instance with a well-separated set of points $P'$ that have positive demands $\{w_j'\}$. In this section, we 
simplify the structure of  the set of ``opened centers'' in the LP solution -- points $u\in P$ with $y_u > 0$. 
Note that $y_u$ may be positive even if $u \notin P'$. We will transform the LP solution $(x,y)$ and obtain a new solution $(x',y')$
with approximately the same LP cost such that $y'_u > 0$ only if $u\in P'$. We will see then that $(x', y')$ is a $(1-\gamma)$-restricted solution of $\cllp$.

Our approach is identical to that in~\cite{charikar2002constant}. If $v \in P\setminus P'$ and $y_v > 0$, we 
move center $v$ to the closest to $v$ point $v'$ in $P'$ by letting $y'_{v'} = y'_{v'} + y_v$. Then we close the center $v$ by letting $y'_v = 0$.
If $v \in P'$, we keep the center at $v$. See Algorithm~\ref{alg:consolidating_centers} for the formal description of this procedure.

\begin{algorithm}[h]
	\begin{algorithmic}[1]
		\STATE {\bfseries Input:} $P', x, y$
		\STATE $x' = x$, $y' = y$
		\FORALL{$v \in P \setminus P'$ with $y'_v >0$}
			\STATE let $v'$ be a closest to $v$ point in $P'$.
			\STATE $y'_{v'} = \min(1, y'_{v'} + y'_v)$, $y'_v = 0$
			\STATE $x'_{u v'} = x'_{u v'} + x'_{uv}$ \text{ and } $x'_{uv} = 0$ \textbf{for all} {$u\in P'$}
		\ENDFOR
	\end{algorithmic}
	\caption{Consolidating centers.}
	\label{alg:consolidating_centers}
\end{algorithm}

Consider a point $u\in P'$ that is fractionally served by center $v\notin P'$ in LP solution $(x,y)$; that is, $x_{uv} > 0$ and $v\notin P$. In the new solution $(x',y')$, it is served by center $v'$. We show that the distance from $u$ to the new center is greater than that to the old one by at most a factor of $2$.

\begin{claim}\label{clm:nn-increase}
Consider a point $v\in P\setminus P'$ and let $v'$ be the nearest neighbor of $v$ in the set $P'$. Then, for every $u\in P'$, $d(u,v') \leq 2 d(u,v)$.
\end{claim}
\begin{proof}
Since (i) $v'$ is a closest point in $P'$ to $v$ and (ii) $u\in P'$, we have $d(v, v') \leq d(v, u)$. Applying the triangle inequality,
we get
$$
d(u,v')
\leq d(u,v) + d(v, v')
\leq 2 d(u,v).
$$
\end{proof}

\begin{claim}\label{clm:local-ball-density}
For each $u\in P'$, $\sum_{v\in \sB(u, r_u)} x_{uv} \geq 1 - \gamma$ where $r_u = \frac{\R(u)}{\gamma^{1/p}}$.
\end{claim}
\begin{proof}
$$\sum_{v\notin \sB(u, r_u)} x_{uv} \leq \sum_{v\notin \sB(u, r_v)} x_{uv} \frac{d(u,v)^p}{r_u^p} \leq \frac{1}{r_u^p} \sum_{v\in P} x_{uv} d(u,v)^p = \frac{\R_u^p}{r_u^p} = \gamma.$$
Thus, $\sum_{v\in \sB(u, r_u)} x_{uv} \geq \sum_{v\in P} x_{uv} - \gamma = 1 - \gamma$.
\end{proof}

\begin{lemma}\label{lem:consolidating_centers}
Algorithm~\ref{alg:consolidating_centers}, given a solution $(x,y)$ for $\cllp(\{w'_j\}, \{\Delta_z(u)\}_{u\in P},\lambda)$ of cost $z'$, constructs a $(1-\gamma)$-restricted w.r.t. $P'$ solution $(x',y')$ for $\cllp(\{w'_j\}, \{\Delta_z(u)\}_{u\in P}, 2\lambda)$ of cost at most $2^p \, z'$.
\end{lemma}
\begin{proof}
Consider solutions $(x, y)$ and $(x', y')$. Let $u\in P'$. From Claim~\ref{clm:local-ball-density} and the inequality $x_{uv} \leq y_v$ (for all $v\in P$), we get that  
$$\sum_{v\in \sB(u, r_u)} y_v \geq 1 - \gamma \qquad \text{where }  r_u = \frac{\R(u)}{\gamma^{1/p}}.$$ 

We now show that for each $v\in \sB(u, r_u)$, point $u$ is the closest neighbor of $v$ in $P'$ and, therefore, Algorithm~\ref{alg:consolidating_centers} reassigns $y_v$ to point $u$. Let $\tilde u$ be a point in $P'$ other than $u$.
Then,
\begin{align*}
d(\tilde u, v)
&\geq d(u,\tilde u) - d(u,v) &&\rhd\text{by the triangle inequality} \\
&> \frac{2}{\gamma^{1/p}} \cdot \R(u) - d(u,v) &&\rhd\text{from Claim~\ref{clm:well-separated}, since } u,\tilde u \in P'  \\
&\geq 2r_u - r_u = r_u \geq d(u, v)&&\rhd\text{since $v\in \sB(u, r_u)$} 
\end{align*}
We conclude that Algorithm~\ref{alg:consolidating_centers} assigns at least $\sum_{\sB(u, r_u)} y_v \geq 1 - \gamma$ to $y_u'$.
We have, $y'_u \geq 1 - \gamma$ for all $u\in P'$ and $y'_u = 0$ for $u\notin P'$. Therefore, $(x',y')$ is a $(1-\gamma)$-restricted solution.

Now we upper bound the cost of $(x', y')$. By Claim~\ref{clm:nn-increase}, when Algorithm~\ref{alg:consolidating_centers} moves a center from point $v\notin P'$ to $v'\in P'$,
the connection cost to $v$ increases by at most a factor of $2^p$: $d(u,v')^p \leq 2^p d(u,v)$ for $u\in P'$. Thus, the cost of $(x',y')$ is at most $2^p\, z'$.

Lastly, we show that $(x', y')$ is a feasible solution of $\cllp(\{w'_j\}, \{\Delta_z(u)\}_{u\in P}, 2 \lambda)$.
Since in each iteration, the value of $\sum_{v\in P} y'_v$ does not increase and the value of $\sum_{v\in P}x'_{uv}$ (for all $u\in P$) does not change, solution $(x',y')$ satisfies constraints~\eqref{cst:point-assignment} and~\eqref{cst:center-count}. 
Also, every time we close a center $v\notin P'$ (see lines 4-6 of Algorithm~\ref{alg:consolidating_centers}),
we increase $x'_{uv'}$ by $x_{uv} \leq y_v$ and increase $y'_{v'}$ by $y_v$ (unless doing so would make $y_v > 1$; in this case, we let $y_v' = 1$). Therefore,
constraint~\eqref{cst:center-cap} remains satisfied throughout the execution of the algorithm.

It now remains to show that $(x',y')$ satisfies constraint~\eqref{const:up-bd} of $\cllp(\{w'_j\}, \{\Delta_z(u)\}_{u\in P}, 2\lambda)$. Consider a pair $u,v'\in P'$ with $x'_{uv'} > 0$. First, assume that $x_{uv'}>0$. Note that $(x,y)$ is a feasible solution of $\cllp(\{w'_j\}, \{\Delta_z(u)\}_{u\in P}, \lambda)$ and thus satisfies constraint~\eqref{const:up-bd}. Therefore, $d(u,v') \leq \lambda \Delta_{z}(u)$, and we are done. 
Now assume that $x_{uv'}=0$. This means that we have closed a fractional center $v\in P\setminus P'$ for $u$ and reassigned $u$ from center $v$ to center $v'$. Then $x_{uv}>0$ and $v'$ is the closest to $v$ point in $P'$. Hence, 
$$
    d(u,v') \leq 2 d(u,v) \leq 2 \lambda \Delta_z(u) 
$$
where the first inequality holds by Claim~\ref{clm:nn-increase}, and the second  inequality holds since $(x,y)$ is a feasible solution for $\cllp(\{w'_j\}, \{\Delta_z(u)\}_{u\in P}, \lambda)$.
We conclude that solution $(x',y')$ is a feasible solution for $\cllp(\{w'_j\}, \{\Delta_z(u)\}_{u\in P}, 2 \lambda)$.
\end{proof} 

We run Algorithm~\ref{alg:consolidating_centers} on the optimal solution
$(x,y)$ for $\cllp(\{w_j\}, \{\Delta_{z_{\g}}(u)\}_{u\in P}, 2)$ and obtain a feasible $(1-\gamma)$-restricted solution $(x',y')$ for $\cllp(\{w_j\}, \{\Delta_{z_{\g}}(u)\}_{u\in P}, 4)$. Since $\cllp(\{w_j\}, \{\Delta_{z_{\g}}(u)\}_{u\in P}, 2)$ is a relaxation for the fair $\ell_p$-clustering,
the cost of $(x,y)$ w.r.t. demands $\{w_j\}$ is at most $z^*$.
By Lemma~\ref{lem:consolidating_locations}, the cost of $(x,y)$ w.r.t 
$\{w_j'\}$ is also at most $z^*$. Finally, by Lemma~\ref{lem:consolidating_centers}, the cost of $(x',y')$
is at most $2^p z^*$

\subsection{Relating the solution costs w.r.t. \texorpdfstring{${\cal I}'$}{I'} and \texorpdfstring{${\cal I}$}{I}}\label{sec:step1-theorem}
In the previous sections, we transformed given instance $\cal I$
to instance ${\cal I}'$ with demands $\{w'_j\}$ and well-separated locations $P'$, then
converted an optimal LP solution $(x,y)$ to a $(1-\gamma)$-restricted solution $(x',y')$.
Now we upper bound the cost of every solution $C$ w.r.t. original demands $\{w_j\}$ in terms of the cost w.r.t. demands $\{w_j'\}$.
\begin{lemma} \label{lem:step1} 
Let $z^*$ be the cost of the optimal solution for $\cal I$. Assume that $z_\g\in [z^*, 2z^*]$. 
 For any integral solution $C'$ with $\faircost(C', \{w'_j\}) \leq z'$,
 $$\faircost(C', \{w_j\}) \leq \frac{2^{2p - 1}}{\gamma} \cdot z^* + 2^{p-1} \cdot z'.$$
\end{lemma}
\begin{proof}
Consider the execution of Algorithm~\ref{alg:consolidating_locations}. The algorithm may either move the demand of point $u$ to some other point $v$ or keep it at $u$.
Let $u' = v$ in the former case and $u' = u$ in the latter case (for every $u\in P$). Note that in either case $d(u, u') \leq \frac{2}{\gamma^{1/p}} \R(u)$.
Therefore,
\begin{align}
d(u, C')^p
&\leq \alpha_p \cdot (d(u, u')^p + d(u', C')^p) &&\rhd\text{Claim~\ref{clm:tri-ineq} (approximate triangle inequality)}\nonumber \\
&\leq \alpha_p \cdot \left(\frac{2^p}{\gamma} \cdot \R(u)^p + d(u', C')^p\right) &&\rhd\text{we upper bound } d(u,u')\label{eq:approximate-triag-ineq-app}
\end{align}
where $\alpha_p = 2^{p-1}$. Hence, for each group $j$,
\begin{align}
\sum_{u\in P} w_j(u) \cdot d(u, C')^p
&\leq \sum_{u\in P} w_j(u) \cdot \left(\frac{2^{2p-1}}{\gamma} \cdot \R(u)^p + 2^{p-1}\cdot d(u', C')^p\right) &&\hspace*{-5mm}\rhd\text{by (\ref{eq:approximate-triag-ineq-app})} \nonumber \\
&\leq \frac{2^{2p-1}}{\gamma} \cdot z^* + 2^{p-1} \cdot \sum_{u\in P} w_j(u) \cdot d(u', C')^p && \hspace*{-5mm}\rhd\text{Claim~\ref{clm:total-R}} \nonumber \\
&= \frac{2^{2p-1}}{\gamma} \cdot z^* + 2^{p-1} \cdot \sum_{v\in P'} w'_j(v) \cdot d(v, C')^p  \nonumber 
\leq \frac{2^{2p-1}}{\gamma} \cdot z^* + 2^{p-1} \cdot z' \nonumber.
\end{align}
It follows that $\faircost(C', \{w_j\}) \leq \frac{2^{2p-1}}{\gamma} \cdot z^* + 2^{p-1} \cdot z'$.
\end{proof}
\section{Step 2: Rounding \texorpdfstring{$(1-\gamma)$}{1-gamma}-Restricted Solutions}\label{sec:rounding}

\subsection{Randomized Rounding for Multiplicative Approximation}
Before describing our randomized rounding procedure, we prove several lemmata, which we will use in the analysis of the algorithm.
\begin{lemma}\label{lem:cost-bound}
Consider an instance of socially fair clustering. Let $\{w'_j\}$ be the set of weights computed by Algorithm~\ref{alg:consolidating_locations} and $(x',y')$ be a feasible solution for $\cllp(\{w'_j\}, z_{\g}, 4)$ computed by Algorithm~\ref{alg:consolidating_centers}. Consider some group $j$ and some point $v\in P'$ with $x'_{vv} < 1$. 
Let $v'$ be a point closest to $v$ in $P'$ other than $v$ itself. Then,
\begin{align*}
w'_{j}(v) \cdot d(v, v')^p \leq \left(2 \cdot 4^p + \frac{8^p}{\gamma}\right)  z^*,
\end{align*}
where $z^*$ is the cost of the optimal solution and $z_\g\in[z^*, 2z^*]$.
\end{lemma}
\begin{proof}
We assume that $w_j'(v) > 0$, as otherwise the statement is trivial.
Let $Q \subseteq P$ be the set of points whose demands have moved to $v$ by Algorithm~\ref{alg:consolidating_locations}. Further, let $Q_1 = Q\cap \sB(v, \Delta_{z_\g}(v) - 0):= \{u: d(u,v) < \Delta_{z_\g}(v)\}$ and $Q_2 = Q\setminus Q_1$. Note that
$$w_j'(v) = \sum_{u\in Q_1}w_j(u) + \sum_{u\in Q_2}w_j(u).$$
We first get an upper bound on $d(v,v')$. Since $x'_{vv} < 1$, we know that $x'_{vu} > 0$ for some $u\in P'$ other than $v$. Constraint~(\ref{const:up-bd}) implies
that $d(v,u) \leq 4 \Delta_{z_{\g}}(v)$ (here we use that $(x',y')$ is a feasible solution for $\cllp(\{w'_j\}, z_{\g}, \mathbf{4})$). Now recall that $v'$ is a closest point to $v$ in $P'$ other than $v$. Thus, $d(v,v') \leq d(v,u) \leq 4 \Delta_{z_\g}(v)$. In particular,
$d(v,v') \leq 4 d(v,q)$ for points $q\in Q_2$. Further, Algorithm~\ref{alg:consolidating_locations} moves demand from $q$ to $v$,
only if $d(v,q) \leq \frac{2}{\gamma^{1/p}} \R(q)$. For $q\in Q_2$, we get
$d(v,v') \leq 4 d(v,q) \leq \frac{8}{\gamma^{1/p}} \R(q)$.
Therefore,
\begin{align*}w'_{j}(v) \cdot d(v, v')^p &= \sum_{u \in Q_1} w_j(u) d(v, v')^p+ \sum_{u \in Q_2} w_j(u) d(v, v')^p \\& 
\leq \sum_{u \in Q_1} w_j(u) (4 \Delta_{z_\g}(v))^p +  \frac{8^p}{\gamma} \sum_{u \in Q_2} w_j(u) \R(u)^p
\end{align*}
Note that $\sum_{u \in Q_1} w_j(u) \Delta_{z_{\g}}(v)^p \le \vol_v(\Delta_{z_{\g}} - 0) \leq z_{\g}$ by the definition of $\Delta_{z_\g}$.
Using this inequality and Claim~\ref{clm:total-R}, we get
$$
w'_{j}(v) \cdot d(v,v')^p \leq 4^p z_\g + \frac{8^p}{\gamma} z^*.
$$
The statement of the lemma follows, since we assume that $z_g \leq 2z^*$.
\end{proof}

Let us now define a forest $F=(P',E)$ on $P'$. We sort all pairs $\{u,v\}$ of distinct points in $P'$ according to the distance between them in ascending order (breaking ties arbitrarily). For every point $u\in P'$, we choose the first pair 
$\{u,v\}$ it appears in and let $u' = v$. Then, $u'$ is a closest point 
to $u$ in $P'$ other than $u$ itself. For every $u$, we add edge $(u, u')$ to 
our graph $F$ (we add every edge at most once). It is easy to see
that the obtained graph is a forest.

\begin{lemma}[cf.~\cite{charikar2002constant}]\label{lem:simple-LP-solution}
Let $(x', y')$ be a feasible $(1-\gamma)$-restricted solution returned by Algorithm~\ref{alg:consolidating_centers} with $\gamma < 1/2$. There exists a feasible solution $(x'',y')$ of cost at most that of $(x',y')$ such that the following holds. For every $v\in P'$, we have
 $x''_{vv} = y'_v \geq 1- \gamma$, $x''_{vv'} = 1 - x''_{vv}$ and $x''_{vu} = 0$ for $u\neq \{v,v'\}$, where $v'$ is as in the definition of $F$.
\end{lemma}
\begin{proof}
We simply let $x''_{vv} = 1 - y'_v$ for every $v\in P'$,
$x''_{vv'} = 1 - y'_v$, and $x''_{vu} = 0$ for $u\notin\{v,v'\}$.
Since $(x,y)$ is a $(1-\gamma)$-restricted solution and $\gamma < 1/2$, $x''_{vu} = 1 -y'_{v} \leq 1 - (1 - \gamma) = \gamma \leq 1 - \gamma \leq y'_u$, as required. It is easy to see that all other LP constraints are also satisfied. Further, 
we chose $x''$ in an optimal way for the given $y'$. 
Thus, the cost of $(x'', y')$ is at most that of $(x',y')$.
\end{proof}
Fix $\gamma = 1/10$. For each $u\in P'$, let $p_u = (1- y'_v)/\gamma$. Note that $0\leq p_u \leq 1$ since $0\leq 1 - y'_v \leq \gamma$. We have
$$\sum_{u\in P'} p_u = \gamma^{-1} \sum_{u\in P'} (1- y'_u) = \gamma^{-1}\left(|P'| - 
\sum_{u\in P'} y'_u\right) \geq \gamma^{-1}(|P'| - k).$$
For every tree $T$ in the forest $F$, do the following. Choose an arbitrary root $r$ in $T$. Partition $T$ into layers based on their depth, starting with the root. Let $A_T$ be the union of every other layer; that is, $A_T$ is the set of vertices of even depth. Finally, let $A$ be the union of all sets $A_T$ over $T$ in $F$. Note that all neighbors of $u\in A$ are not in $A$; all neighbors of $u\notin A$ are in $A$.

If $\sum_{u\in A} p_u \geq \frac{1}{2\gamma}(|P'| - k)$, let $S = A$; otherwise, let $S = P'\setminus A$. In either case, $\sum_{u\in S} p_u \geq \frac{1}{2\gamma}(|P'| - k)$.
Finally, we construct our combinatorial solution $C$:
\begin{itemize}
\item if $|P'| \leq k$, let $C= |P'|$; otherwise, proceed as follows
\item add all points from $P'\setminus S$ to $C$,
\item add each point $v \in S$ to $C$ with probability $1-p_v$ (independently).
\end{itemize}
We show now that $|C| \leq k$ with probability at least $3/4$
and $\faircost(C) \leq O\left(2^{3p}\right)\frac{\log \ell}{\log\log \ell} \cdot z^*$ (where $z^*$ is the cost of the optimal solution) with probability at least $1 - 1/(2\ell)$.

\begin{theorem}\label{thm:rounding}
With probability at least $1/4$, $|C| \leq k$ and $\faircost(C, \{w_j'\}) = O(2^{3p} \cdot \frac{\log\ell}{\log \log\ell} z^*)$.
\end{theorem}
\begin{proof}
If $|P'| \leq k$, then $|C| = |P'| \leq k$ and the cost of $C$ w.r.t. demands $\{w_j'\}$ is 0; thus, the statement of the theorem trivially holds. We assume below that $|P'|  > k$.
By Lemmata~\ref{lem:consolidating_centers} and~\ref{lem:simple-LP-solution},
$(x', y')$ and $(x'', y')$ are feasible solutions for 
$\cllp(\{w'_j\}, z_\g, 4)$. Let $z$ be the LP cost of $(x'',y')$ w.r.t. demands $\{w_j'\}$. Then $z \leq 2^p z^*$.

First, we bound the size of $C$. Let $X_v$ be the indicator random variable of the event $v\notin C$; i.e., $X_v =1$ if $v\notin C$ and $0$ otherwise. Define $X = \sum_{v\in P'} X_v$. Now we are lower bounding the number $X$ of points in $P'$
that are \textit{not} centers in $C$.
\begin{align*}
\E[X] =\E\left[\sum_{v\in P'} X_v\right] = \sum_{v\in S} \E[X_v] + \sum_{v\in P'\setminus S} \E[X_v] = \sum_{v\in S} p_v \geq \frac{|P'| - k}{2\gamma} .
\end{align*}
Applying the Chernoff bound, we get for $\varepsilon = 1 - 2\gamma$,
\begin{align*}
\Pr( X \leq |P'| - k) \leq \Pr(X \leq (1 - \varepsilon)\E(X)) \leq \exp(-\left.\varepsilon^2 \E[X]\right/2)
\end{align*}
Note that  $\E[X] \geq \frac{1}{2\gamma} = 5$ and $\varepsilon \geq 4/5$, since $|P'| \geq k + 1$ and $\gamma =1/10$. Thus, $\Pr[X \leq  (1 - \varepsilon)\E(X)] < 1/4$. Hence, with probability at least $3/4$,
\begin{align}\label{eq:center-size}
|C| \leq |P'| - X \leq |P'| -  (|P'| - k) =  k.
\end{align}

Fix $j\in [\ell]$. Now we show that $\cost(C, w_j')$ is at most $O\left(2^{3p-1} \log \ell / \log\log \ell\right) z^*$ with high probability. Then applying the union bound, we will get the same bound for all $j\in[\ell]$ and thus for $\faircost(C, \{w_j'\})$.

For every $v\in P'$, let $Y_v = w'_j(v) \cdot d(v, C)^p$. Note that for every $v\in P'$ either $v\notin S$ or $v'\notin S$ (since $v$ and $v'$ are neighbors in $F$).
Therefore, we always have that at least one of the points $v$ and $v'$ is in $C$. Further, if $X_v = 0$, then $v\in C$, and thus $d(v, C) = 0$;
if $X_v = 1$, then $d(v, C) = d(v,v')$.
Hence,
$Y_v = w'_j(v) d(v,v')^p X_v$.
From Lemma~\ref{lem:cost-bound}, we get for each $v\in P'$,
\begin{align*}
Y_v \leq w'_j(v) \cdot d(v, v')^p \leq \left(2 \cdot 2^{2p}+ \frac{2^{3p}}{\gamma}\right) z^*= O\left(2^{3p} z^*\right)\quad\text{(always)}.
\end{align*}
Let $Z_j : = \sum_{v\in P_j} Y_v = \sum_{v\in S} Y_v$ be the cost of group $j$ w.r.t.~weights $\{w_j'\}$. Note that all random variables $\{Y_v\}_{v\in P_j}$ are independent. Random variables $Y_u$ for $u\notin S$ are identically equal to 0. We have,
\begin{align*}
\E[Z_j] \leq \sum_{v\in S} w'_j(v) d(v, v')^p p_v = \frac{1}{\gamma} \sum_{v\in S} w'_j(v) d(v, v')^p \underbrace{(1-y'_v)}_{x''_{vv'}} \leq \frac{z}{\gamma} \leq \frac{2^p}{\gamma} z^* 
\end{align*}
(recall that $z$ is the cost of the LP solution $(x'', y')$)
and 
\begin{align*}
\Var[Z_j]
= \sum_{v\in S} \Var[Y_v] 
\leq \sum_{v\in S} (w'_j(v) d(v, v')^p)^2\, p_v 
\stackrel{\text{\tiny{by Lemma~\ref{lem:cost-bound}}}}{\leq} O\left(2^{3p}z^*\right) \cdot \E[Z_j] = O\left(2^{4p} ({z^*})^2\right). 
\end{align*}

Now we will use Bennett's inequality (Theorem~\ref{thm:Bennett}) 
to bound random variables $Y_j$.
To do so, we define \textit{zero-mean} versions of random variables $Y_v$ and $Z_j$. Let $Y'_v := Y_v - \E[Y_v]$ and $Z'_j := \sum_{v\in P'_j} Y'_v$.
Note that, $|Y'_v| = O(2^{3p}z^*)$ (always) and $\Var[Z'_j] = \Var[Z_j] = O\left(2^{4p} ({z^*})^2\right)$.
Applying Bennett's inequality to  $Z'_j$, we get for $\tau > 0$
\begin{align*}
\Pr\left(Z_j' \geq \tau z^*\right)
\leq \exp\left(-\Omega\left(\frac{\tau}{2^{3p}}\right) \log\left(\Omega\left(\frac{\tau}{2^p}\right) + 1\right)\right)
\end{align*}
Letting $\tau = c \cdot 2^{3p} \frac{\log \ell}{\log\log \ell}$
for large enough $c$, we get that $\Pr\left(Z_j' \geq \tau z^*\right) \leq \frac{1}{2\ell^2}$, which implies that  $\Pr\left(Z_j \geq c' \cdot 2^{3p} \frac{\log \ell}{\log\log \ell} z^*\right) \leq \frac{1}{2\ell^2}$ for some absolute constant $c'$.
Applying the union bound, we get that with probability at least $1 -\frac{1}{2\ell}$, the following upper bound on the fair cost holds
\begin{equation} \label{eq:main}
\faircost(C, \{w_j'\})=
\max_{j\in [\ell]} \sum_{v\in P'} w'_{j}(v)\cdot d(v, C)^p = O\left(2^{3p} \cdot \frac{\log \ell}{\log \log \ell} z^* \right).
\end{equation}

We conclude that with probability at least $1 - \frac{1}{2\ell} - 1/4 \geq \frac{1}{4}$, both $|C| \leq k$ and (\ref{eq:main})
holds.
\end{proof}

\begin{algorithm}[h]
	\begin{algorithmic}[1]
		\STATE {\bfseries Input:} An instance $\{w_j\}$ of socially fair clustering, $z_\g\in [z^*, 2z^*]$, desired error probability $\varepsilon$
		\STATE $\gamma = 1/10$, $\lambda = 2$
		\STATE Solve $\cllp(\{w_j\}, z_\g, 2)$. Let $(x,y)$ be an optimal fractional solution.
		\STATE Run Algorithm~\ref{alg:consolidating_locations} and obtain new demands $(\{w'_j\})$
		and set $P'$
    	\STATE Run Algorithm~\ref{alg:consolidating_centers} and obtain LP solution $(x',y')$
    	\STATE Define $x''$ as in Lemma~\ref{lem:simple-LP-solution}.
    	\STATE Construct forest $F$ and find set $S$.

			\STATE Run randomized rounding $\lceil \log_{4/3} (1/\varepsilon)\rceil$ times
			\STATE Let $C$ be the best of the solutions (with at most $k$ centers) that randomized rounding generates.
     \RETURN C
	\end{algorithmic}
	\caption{$(e^{O(p)}\frac{\log \ell}{\log \log \ell})$-approximation algorithm for socially fair $\ell_p$-clustering}
	\label{alg:main}
\end{algorithm}

\begin{proof}{\bf of Theorem~\ref{thm:main-approx}}
We put together all the steps we described in this paper.
The entire algorithm is shown as Algorithm~\ref{alg:main}.
As discussed in Appendix~\ref{sec:appendix-guess-z}, we may assume that $z_\g\in [z^*, 2z^*]$ is given to us.
By Theorem~\ref{thm:rounding}, the probability that our randomized rounding procedure will find a feasible solution $C$ for demands $\{w_j'\}$ of cost at most $O\left(2^{3p}\right)\frac{\log \ell}{\log \log \ell}\cdot z^*$ is at least $1/4$.
Since we run randomized rounding $\lceil \log_{4/3} (1/\varepsilon)\rceil$ times, we will succeed at least once with probability at least $1-\varepsilon$.
By Lemma~\ref{lem:step1}, the cost of $C$ w.r.t. the original demands $\{w_j\}$ is upper bounded as follows,
\begin{align*}
    \faircost(C,\{w_j\}) 
    \leq \frac{2^{2p-1}}{\gamma} \cdot z^* + 2^{p-1} \cdot \faircost(C, \{w'_j\}) 
    = O\left(2^{4p}\right)\frac{\log \ell}{\log \log \ell} z^*
\end{align*}
\end{proof}

\section{Deterministic Rounding for Bicriteria Approximation}
\label{sec:bicriteria}
In this section we show how to get our bicriteria approximation algorithm.

\begin{proof}{\bf of Theorem~\ref{thm:bicriteria-approx}}
We solve LP relaxation $\cllp(\{w_j\}, \{\Delta_{z_\g}(u)\}_{u\in P}, 2)$ and then run Algorithms~\ref{alg:consolidating_locations} and~\ref{alg:consolidating_centers}. We obtain a fractional solution $(x', y')$ and set $P'$.
By Lemma~\ref{lem:consolidating_locations}, $(x', y')$ is a $(1-\gamma)$-restricted solution for  $\cllp(\{w'_j\}, \{\Delta_{z_\g}(u)\}_{u\in P}, 4)$ w.r.t.~$P'$ of cost at most $2^p z^*$.
We return set $C = P'$.

By LP constraint~(\ref{cst:center-count}), we have $\sum_{u\in P} y_u \leq k$. Now, since $(x', y')$ is a $(1-\gamma)$-restricted solution,
$y_u \geq 1 - \gamma$ for $u\in P'$.
 Therefore, $|C| = |P'| \leq k/(1 - \gamma)$, as required.
Note that for every $u\in P'$, we have $d(u,C)^p = 0$. Therefore,
\begin{align*}
\faircost(C, \{w'_j\})
= 0.
\end{align*}
Now, we apply Lemma~\ref{lem:step1} to bound $\faircost(C, \{w_j\})$,
\begin{align*}
\faircost(C, \{w_j\})
\leq \frac{2^{2 p-1}}{\gamma} \cdot z^* + 2^{p-1}\cdot \faircost(C, \{w'_j\})= 
\frac{2^{2p-1}}{\gamma} z^* .
\end{align*}
\end{proof}
\acks{We thank Viswanath Nagarajan for bringing papers 
\citep{anthony2010plant} and \citep{bhattacharya2014new}
to our attention. YM was supported by NSF awards CCF-1718820, CCF-1955173, and CCF-1934843. AV was supported by NSF award CCF-1934843.}

\bibliography{ref-cls}
\bibliographystyle{abbrvnat}
\appendix
\section{Finding the value of \texorpdfstring{$z_\g$}{z}}\label{sec:appendix-guess-z}
In this section, we explain why we may assume that we know $z_\g\in [z^*,  2z^*]$. 
\begin{observation}
There exist $u, v\in P$ and $j\in[\ell]$ such that
$$w_j(u)\cdot d(u,v)^p \leq z^* \leq n\cdot w_j(u) \cdot d(u,v)^p$$
\end{observation}
\begin{proof}
Let $C^*$ be an optimal solution.
Assume that group $j$ has the largest cost: $\cost(C^*, w_j) \geq \cost(C^*, w_i)$ for all $i\in[\ell]$. Then,  $z^* = \cost(C^*, w_j)$.
Consider $u\in P_j$ that maximizes $w_j(u)\cdot d(u,C^*)^p$. 
Let $v \in C^*$ be the closest center to $u$. Note that 
$w_j(u)\cdot d(u,v)^p = w_j(u)\cdot d(u,C^*)^p$
and
$$w_j(u)\cdot d(u,v)^p 
\leq \cost(C^*, w_j)
\leq |P_j|\cdot w_j(u) \cdot d(u,v)^p
\leq n\cdot w_j(u) \cdot d(u,v)^p.
$$
\end{proof}

Now we assume that we have an algorithm that with high probability, returns a $\beta$-approximate solution for socially fair $\ell_p$-clustering when $z^* \leq z_\g \leq 2z^*$.
We run this $\beta$-approximation algorithm with different estimates for $z_\g$ of the form $2^i w_j(u) \,d(u,v)^p$ where $i\in\{0,\cdots,  \lfloor\log_2 n\rfloor\}$, $u,v\in P$,
and $j\in [\ell]$. We output the best solution we find.
Note that our algorithm runs in polynomial time, since it invokes the $\beta$-approximation
algorithm at most $O(n^2\ell \log n)$ times.
We note this step can often be significantly sped up; in particular, if the cost is represented as a floating-point number with $D$ binary digits, we can run the $\beta$-approximation algorithm at most $D$ times.
\section{Integrality Gap for Relaxation \texorpdfstring{$\cllp(\{w_j\}, \{\Delta_z(u)\}_{u\in P}, 2)$}{ClusterLP}}\label{sec:int_gap}

In this section, we define what the integrality gap for $\cllp(\{w_j\}, \{\Delta_z(u)\}_{u\in P}, 2)$ is (this is not necessarily straightforward, since the LP depends on $z$) and show that the gap is $\Omega(\log \ell/\log\log\ell)$. 

\paragraph{Main approximation result restated.} The main result of our paper can be formulated as follows (which is implicit in the proof of Theorem~\ref{thm:main-approx}). Let $f(z)$ be the LP cost of the optimal solution for $\cllp(\{w_j\}, \{\Delta_z(u)\}_{u\in P}, 2)$ and $g(z) = \max(z, f(z))$. Then the following items hold.
\begin{enumerate}
    \item For every $z$, there exists a combinatorial solution of cost at most $e^{O(p)} \frac{\log\ell}{\log\log\ell}\, g(z)$. Further, this solution can be found in polynomial time.
    \item For $z \geq z^*$, $g(z) = z$.
\end{enumerate}
To obtain our result, we take $z_\g\in [z^*, 2z^*]$. By item 2, $g(z_g) = z_g \leq 2z^*$. By item 1, we can find a solution of cost $e^{O(p)} \frac{\log\ell}{\log\log\ell} g(z_\g) \leq e^{O(p)} \frac{\log\ell}{\log\log\ell} z^*$.

\paragraph{Integrality gap.} We now show that dependence on $\ell$ in item $1$ cannot be improved. Namely, we prove that there exists a sequence of instances and parameters $z$ such that every combinatorial solution has cost at least $\Omega\left(\frac{\log\ell}{\log\log \ell}g(z)\right)$.

We construct the following instance for every $k\geq 1$.
Let $t = \floor{\sqrt{k}}$ and $n=k + t$. Consider a metric space $P$ on $n$ points, in which the distance between every two distinct points is $1$.
For every set $A$ of $t$ points, create a group $P_A = A$. Note that the total number of groups is $\ell =\binom{n}{t} = e^{\Theta(\sqrt{k}\log k})$. 

Let $z = 1$. Then $\Delta_z(u) = 1$ for every $u$. Note
that constraint~\ref{const:up-bd} is trivially satisfied by any LP solution, since there are simply no two points $u$ and $v$ with $d(u,v) > 2\Delta_z(v) = 2$.
Observe that the cost of every solution $C$ of size $k$ is $t$, since the cost of $C$ for group $P_{P\setminus C} = P\setminus C$ is
$$\sum_{u\in P\setminus C} d(u, C)^p = \sum_{u\in P\setminus C} 1 = t$$
and the cost of every group $P_A$ is at most $|P_A| =t$.
Now we construct an LP solution of cost at most $1$.
For every $u$, we let $y_u = k/n$, $x_{uu} = k/n$, $x_{uv_u} = 1 - k/n = t/n$ for an arbitrary $v_u\neq u$,
and $x_{uv'} = 0$ for $v'\notin \{u,v_u\}$. It is immediate that this is a feasible LP solution. Its cost
is the maximum over all groups $P_A$ of
$$\sum_{u\in P_A} \left(\frac{k}{n} \cdot d(u,u)^p + \frac{t}{n} \cdot  d(u, v_u)^p\right) = \sum_{u\in P_A} \frac{t}{n} = t^2/n \leq 1.$$
We conclude that $f(z) \leq 1$ and $g(z) = \max(z, f(z)) = 1$.
Therefore, the integrality gap of this instance is 
$$t / 1 = \floor{\sqrt{k}} = \Theta\left(\frac{\log \ell}{\log\log \ell}\right).$$
\end{document}